\documentclass[9pt]{acm_proc_article-sp}

\usepackage{color}
\usepackage{algorithm} 
\usepackage{algorithmic} 
\usepackage{graphicx}
\usepackage{amsmath}
\usepackage{amssymb}
\usepackage{bbm}

\begin{document}

\newcommand{\textblue}[1]{\textcolor{blue}{#1}}
\newcommand{\todo}[1]{\textblue{TODO: #1}}
\newtheorem{theorem}{Theorem}[section]
\newtheorem{corollary}{Corollary}[section]
\newcommand{\csp}[0]{CSP}


\newcommand{\fullpaperappendix}[0]{}

\pagenumbering{arabic}

\title{Exploring Privacy Preservation in Outsourced K-Nearest Neighbors with Multiple Data Owners}

\def\icsi{\raisebox{6pt}{\small $\ast$}}
\def\ucb{\raisebox{6pt}{\small $\dagger$}}

\newcommand{\authspace}{\hspace{5mm}}

\numberofauthors{1}

\author{
  \alignauthor
  Frank Li\ucb \authspace
  Richard Shin\ucb \authspace
  Vern Paxson\ucb\icsi \authspace\\
  \vspace{.2cm}
  \affaddr{
    \ucb University of California, Berkeley \authspace
    \icsi International Computer Science Institute
  }\\
  \vspace{.2cm}
  \affaddr{
    \{frankli, ricshin, vern\}@cs.berkeley.edu
  }
  \vspace{.3cm}
}

\maketitle
\begin{abstract}
The k-nearest neighbors (k-NN) algorithm is a popular and effective classification algorithm. Due to its large storage and computational requirements, it is suitable for cloud outsourcing. However, k-NN is often run on sensitive data such as medical records, user images, or personal information. It is important to protect the privacy of data in an outsourced k-NN system.

Prior works have all assumed the data owners (who submit data to the outsourced k-NN system) are a single trusted party. However, we observe that in many practical scenarios, there may be multiple mutually distrusting data owners. In this work, we present the first framing and exploration of privacy preservation in an outsourced k-NN system with multiple data owners. We consider the various threat models introduced by this modification. We discover that under a particularly practical threat model that covers numerous scenarios, there exists a set of adaptive attacks that breach the data privacy of any exact k-NN system. The vulnerability is a result of the mathematical properties of k-NN and its output. Thus, we propose a privacy-preserving alternative system supporting kernel density estimation using a Gaussian kernel, a classification algorithm from the same family as k-NN. In many applications, this similar algorithm serves as a good substitute for k-NN. We additionally investigate solutions for other threat models, often through extensions on prior single data owner systems.
\end{abstract}

\pagenumbering{arabic}

\section{Introduction}

The k-nearest neighbors (k-NN) classification algorithm has been widely and effectively used for machine learning applications. Wu \textit{et al.} categorized it as one of the top 10 most influential data mining algorithms \cite{Wu:2007:TAD:1327434.1327436}. K-NN identifies the $k$ points nearest to a query point in a given data set, and classifies the query based on the classifications of the neighboring points. The intuition is that nearby points are of similar classes. K-NN classification benefits from running over large data sets and its computation can be expensive. These characteristics make it suitable to outsource the classification computation to the cloud. However, k-NN data is often sensitive in nature. For example, k-NN can be applied to medical patient records, census information, and  facial images. 
It is critical to protect the privacy of such data when outsourcing computation.

The outsourced k-NN model focuses on providing classification, rather than training. Hence, it is assumed that the algorithm parameters have been adequately chosen through initial investigation. Prior work on preserving privacy in such a model uses computation over encrypted data \cite{wong, no-share-key, arxiv-knn, knn-revisted}. These works have all assumed a single trusted data owner who encrypts her data before sending it to a cloud party.
Then queriers can submit encrypted queries to the system for k-NN classification. In this setting, we would like to keep data private from the cloud party and queriers, and keep queries private from the cloud party and the data owner. In some existing works, the queriers share the secret (often symmetric) data encryption keys. In others, the querier interacts with the data owner to derive the encryption for a query without revealing it. 

From these existing models, we make a simple observation with non-trivial consequences: the data owner may not be completely trusted. In all previous models, trust in the data owner is natural since only the owner's data privacy is at risk. However, this assumption does not hold in some practical scenarios:
\begin{itemize}
\item Multiple mutually-distrusting parties wish to aggregate their data for k-NN outsourcing, without revealing data to each other. Since k-NN can perform significantly better on more data, all parties can benefit from improved accuracy through sharing. As an example, hospitals might contribute medical data for a k-NN disease classification study or a service available to doctors. However, they do not want to release their patient data in the clear to each other or a cloud service provider. In some cases, these data owners may act adversarially if that allows them to learn the other owners' data.

\item The k-NN outsourced system allows anyone to be a data owner and/or querier. This is a generalization of the previous scenario, and a privacy-preserving solution allows data owner individuals to contribute or participate in a system directly without trusting any other parties with plaintext data. Some potential applications include sensor information, personal images, and location data.
\end{itemize}

These scenarios involve data aggregated from multiple owners, hence we term this the \textit{multi-data owner outsourced} model. In this paper, we provide the first framing and exploration of privacy preservation under this practical model. Since these multi-data owner systems have not yet been used in practice (perhaps because no prior works address privacy preservation), we enumerate and investigate variants of a privacy threat model. However, we focus on one variant that we argue is particularly practical and covers a number of interesting cases. 
Under this model, we discover a set of adaptive privacy-breaching attacks based purely on the nature of how k-NN works, regardless of any system protocol and encryption designs.  

To counter these privacy attacks, we propose using kernel density estimation with a Gaussian kernel instead of k-NN. This is an alternative algorithm from the same class as k-NN (which we can consider as kernel density estimation with a uniform kernel). While this algorithm and k-NN are not equivalent, we demonstrate that in many applications the Gaussian kernel should provide similar accuracy.
We construct a privacy-preserving scheme to support such an algorithm using partially homomorphic encryption and garbled circuits, although at a computational and network cost linear to the size of the data set. Do note that k-NN itself is computationally linear. While this system may not yet be practical for large data sets, it is both proof of the existence of a theoretically secure alternative as well as an first step to guide future improvements.

In summary, our contributions are:
\begin{itemize}
\item We provide the first framework for the k-NN multi-data owner outsourced model.
\item We explore privacy preservation of k-NN under various threat models. However, we focus on one threat model we argue is both practical and covers several realistic scenarios. We discover a set of privacy-breaching attacks under such a model.
\item We propose using kernel density estimation with a Gaussian kernel in place of k-NN. We describe a privacy-preserving construction of such a system, using partially homomorphic encryption and garbled circuits.
\end{itemize}

\section{Background}

\subsection{Nearest Neighbors}
\label{sec:background-nn}

\sloppy K-nearest neighbors (k-NN) is a simple yet powerful non-parametric classification algorithm that operates on the intuition that neighboring data points are similar and likely share the same classification. It runs on a training set of data points with known classifications: $\mathcal{D} = \{(\mathbf{x}_1, y_1)$, $\cdots,$ $(\mathbf{x}_n, y_n)\}$, where $\mathbf{x}_i$ is the $i$-th data point and the label $y_i \in \{1, \cdots, C\}$ indicates which of the $C$ classes $\mathbf{x}_i$ belongs to. For a query point $\mathbf{x}$, k-NN determines its classification label $y$ as the most common label amongst the $k$ closest points to $\mathbf{x}$ in $\mathcal{D}$. Closeness is defined under some distance metric (such as Manhattan or Euclidean distance for real-valued vectors). Since all of k-NN's computation is at classification time, instead of training time, $k$-NN is an example of a lazy learning method.


In the case of 1-NN ($k=1$), as the number of training data points approaches infinity, the classification error of k-NN becomes bounded above by twice the Bayes error (which is the theoretical lower bound on classification error for any ideal classification algorithm). 
In general, k-NN benefits greatly from execution over large data sets. Given a single data owner will have a limited amount of data, aggregation of data from multiple parties is desirable for increased accuracy in k-NN applications. Also, since k-NN classification is computationally expensive, it is practical to outsource computation. These characteristics motivate our interest in multi-data owner outsourced k-NN classification.

\subsection{Kernel density estimation and regression}
\label{sec:background-kde}

K-NN is an example of an \emph{instance-based} learning method as it involves making a prediction directly from existing instances in the training data (as opposed to using a model of the training data, for example).
K-NN uses only the $k$ nearest instances to a query and allows each to equally influence the query classification. However, there are other possible rules to decide how much influence each instance should have on the final prediction, known as \emph{kernels}.\footnote{To disambiguate between the many distinct meanings of ``kernel'' in machine learning, statistics, and computer science, the kind of kernels in this paper are also called ``smoothing kernels''.}

Formally, a kernel $K(x)$ is a function which satisfies
\begin{align*}
  (1)\hspace{1mm} \int_{-\infty}^\infty K(x) dx &= 1 \hspace{7mm} (2)\hspace{1mm} K(x) = K(-x).
\end{align*}
In other words, it integrates to $1$ and is symmetric across $0$.

Given $n$ samples $\{\mathbf{x}_1, \cdots, \mathbf{x}_n\}$ of a random variable $\mathbf{x}$, one can estimate the probability density function $p(\mathbf{x})$ with a kernel $K$ in the following way, called kernel density estimation:
\begin{align*}
  p(\mathbf{x}) &= \frac{1}{n} \sum_{i=1}^n K(\|\mathbf{x} - \mathbf{x}_i\|)
\end{align*}
where $\|\cdot\|$ is a norm. Given this estimate, classification can be determined as the most probable class:
\begin{align*}
  \mathcal{D}_C &= \{x_i\,|\,(x_i, y_i) \in \mathcal{D}, y_i = C\} \\
  p(y = C | \mathbf{x}, \mathcal{D}) 
	&\propto \sum_{\mathbf{x}_i \in \mathcal{D}_C} K(\|\mathbf{x} - \mathbf{x}_i\|) \\
	\arg\max_C p(y = C | \mathbf{x}, \mathcal{D}) &= \arg\max_C \sum_{\mathbf{x}_i \in \mathcal{D}_C} K(\|\mathbf{x} - \mathbf{x}_i\|)
\end{align*}
See Appendix~\ref{sec:background-kde-appendix} \fullpaperappendix for a more detailed exposition.

Therefore, to classify a particular point $\mathbf{x}$, we can sum the kernel values of the points which belong to each class and determine which class has the largest sum.
The following uniform kernel equates this derivation with k-NN:
\begin{align*}
  d_{k, \mathbf{x}, \mathcal{D}} &:= \text{distance of $k$\textsuperscript{th} nearest point from $\mathbf{x}$ in $\mathcal{D}$} \\
  K(\|\mathbf{x} - \mathbf{x}_i\|; k, \mathcal{D}) &= \begin{cases}
    \frac{1}{2d_{k, \mathbf{x}, \mathcal{D}}} & \text{if $\|\mathbf{x} - \mathbf{x}_i\| \le d_{k, \mathbf{x}, \mathcal{D}}$} \\
    0 & \text{otherwise}
  \end{cases}
\end{align*}
$d_{k, \mathbf{x}, \mathcal{D}}$ is the \emph{width} of this kernel, as $K(t) > 0$ for $||t-\mathbf{x}|| \le d_{k, \mathbf{x}, \mathcal{D}}$ and $K(t) = 0$ for $||t-\mathbf{x}|| > d_{k, \mathbf{x}, \mathcal{D}}$. Thus, k-NN classification is kernel density estimation with a uniform finite-width kernel.

One can substitute in a different kernel to obtain a classifier which behaves similarly. One example is the Gaussian kernel:
\begin{align*}
  K(x) = \frac{1}{\sigma \sqrt{2 \pi}} e^{-\frac{x^2}{2 \sigma^2}}
\end{align*}
where $\sigma$ is the standard deviation parameter. Note that this kernel has infinite width, meaning all points have some influence on the classification.

\subsection{Paillier Cryptosystem}
\label{sec:background-paillier}

The Paillier cryptosystem \cite{paillier} is a public-key semantically secure cryptographic scheme with partially homomorphic properties. These homomorphic properties allow certain mathematical operations to be conducted over encrypted data.
Let the public key $pk$ be $(N, g)$, where $N$ is a product of two large primes and $g$ is a generator in $\mathbb{Z}_{N^2}^{*}$. The secret decryption key is $sk$. Also let $E_{pk}$ be the Paillier encryption function, and $E^{-1}_{sk}$ be the decryption function. Given $a,b \in \mathbb{Z}_N$, the Paillier cryptosystem has the following properties:
\begin{align*}
  E_{sk}^{-1}(E_{pk}(a) \cdot E_{pk}(b) \mod N^2) &= a+b \mod N \tag{Add} \\
  E_{sk}^{-1}(E_{pk}(a)^b \mod N^2) &= a \cdot b \mod N \tag{Mult}
\end{align*}
In other words, multiplying the ciphertext for two values results in the ciphertext for the sum of those values, and computing the $c$-th power of a value's ciphertext will result in the ciphertext for the value multiplied by $c$.

The Paillier ciphertext is twice the size of a plaintext; if $N$ is 1024 bits, the  ciphertext is 2048 bits. A micro-benchmark of Paillier in \cite{cryptdb} shows practical performance runtimes: encryption of a 32-bit integer takes 9.7 ms, decryption takes 0.7 ms, and the homomorphic addition operation takes 0.005 ms.

\subsection{Yao's Garbled Circuits}

Yao's garbled circuit protocol \cite{yao, yao_proof} allows a two-party evaluation of a function $f(i_1, i_2)$ run over inputs from both parties, without revealing the input values when assuming a semi-honest adversary model. If Alice and Bob have private input values $i_A$ and $i_B$, respectively, the protocol is run between them (the input owners) and outputs $f(i_A, i_B)$ without revealing the inputs to any party. 

In the protocol, the first party is called the generator and the second party is the evaluator. The generator takes a Boolean circuit for computing $f$, and generates a ``garbled" version $GF$ (intuitively, it is cryptographically obfuscating the circuit logic). Any input $i$ for function $f$ has a mapping to a garbled input for $GF$, which we will denote as $GI(i)$. The generator gives the garbled circuit $GF$ to the evaluator, as well as the generator's garbled inputs $GI(i_1)$. Since the generator created the garbled circuit, only the generator knows the valid garbled inputs. The evalutor then engages in a 1-out-of-2 oblivious transfer protocol \cite{oblivious_transfer, 1_2_ot} with the generator to obliviously obtain the garbled input values for her own private inputs $i_2$. The evaluator can now evaluate $GF(GI(i_1), GI(i_2))$ to obtain a garbled output, which maps back to the output of $f(i_1, i_2)$.

For a more concrete understanding of the garbled circuit itself, consider a single binary gate $g$ (e.g., an AND-gate) with inputs $i$ and $j$, and output $k$. For each input and the output, the generator generates two random cryptographic keys $K_x^0$ and $K_x^1$ that correspond to the bit value of $x$ as 0 and 1, respectively. The generator then computes the following four ciphertexts using a symmetric encryption algorithm $Enc$ (which must be IND-CPA, or indistinguishable under chosen-plaintext attacks):
\begin{align*} 
Enc_{(K_i^{b_i}, K_j^{b_j})} (K_k^{g(b_i,b_j)})\text{ for } b_i, b_j \in \{0,1\}
\end{align*} 
A random ordering of these four ciphertexts represents the garbled gate.  Knowing $(K_i^{b_i}, K_j^{b_j})$ allows the decryption of only one of these ciphertexts to yield the value of $K_k^{g(b_i,b_j)}$ and none of the other outputs. Similarly, valid outputs cannot be obtained without the keys associated with the gate inputs. A garbled circuit is the garbling of all the gates in the Boolean circuit for $f$, and can be evaluated gate-by-gate without leaking information about intermediate computations. 
There exists efficient implementations of garbled circuits \cite{yao_implementation, aby, efficient_garbling, tinygarble}, although naturally the garbled circuit's size and evaluation runtime increases with the complexity of the evaluated function's circuit.

\section{Related Work}
\label{sec:related}

\begin{figure}
\centering
	\includegraphics[width=0.4\textwidth, height=35mm, keepaspectratio=true]{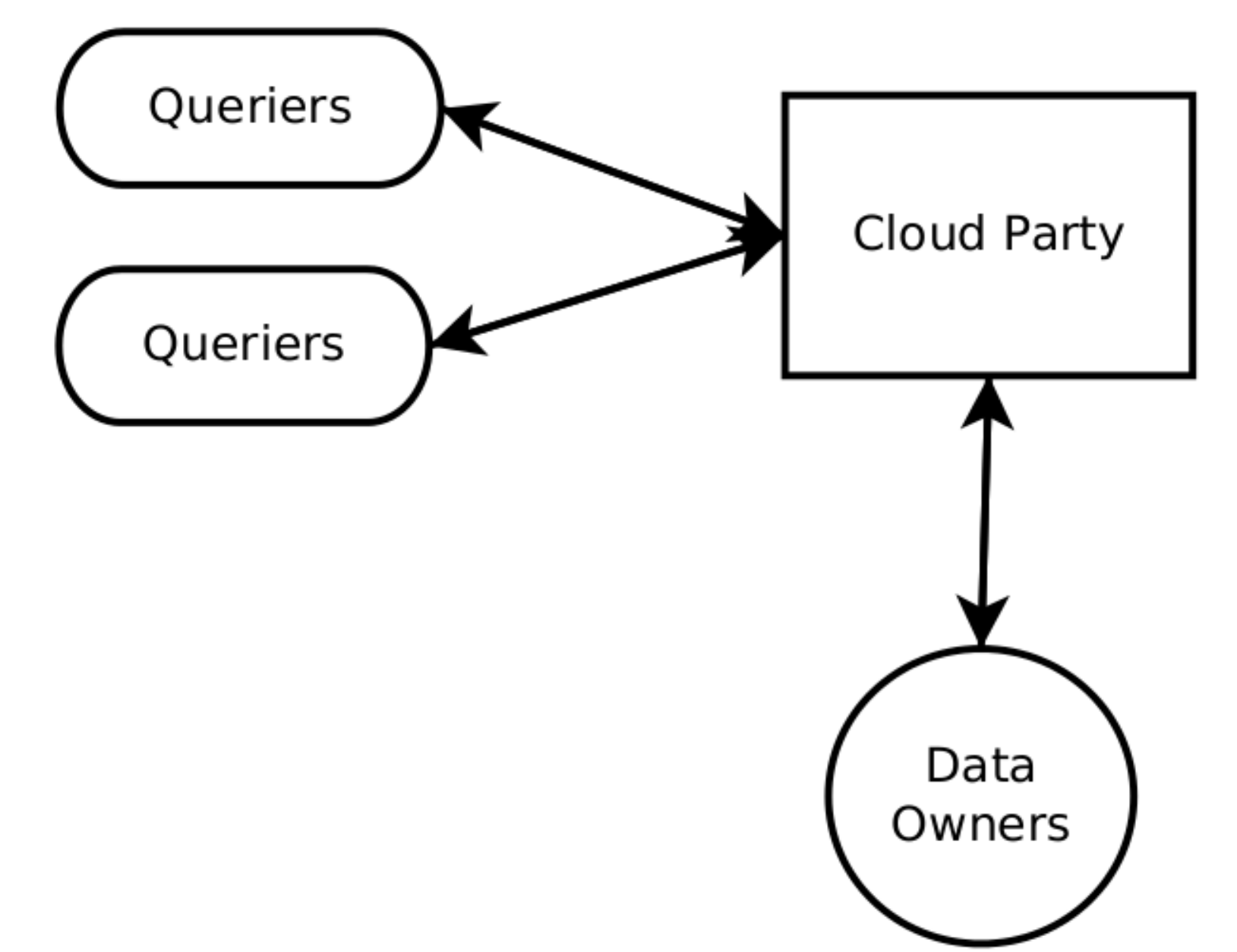}
	\caption{The general model of existing outsourced k-NN system. The data owner is trusted and outsources encrypted data to the cloud party. Queriers request k-NN computation from the cloud party, and are sometimes trusted depending on the prior work. The cloud party executes k-NN classification and is semi-honest.}
	\label{fig:systemold}
\end{figure}

Prior works have focused on two general approaches to achieving privacy in k-NN: distributed k-NN and outsourced k-NN computed on encrypted data.

In distributed k-NN, multiple parties each maintain their own data set. Distributed computation involves interactions between these parties to jointly compute k-NN without revealing other data values \cite{mpc-kNN, distributed-1, distributed-2, distributed-3,
distributed-4, distributed-6, aggregation-outsourcing, 10.1109/ICDMW.2006.3}. A general framework for how these systems operate is they iteratively reveal the next closest neighbor until $k$ neighbors have been determined. While allowing multiple parties to include their data sets, these works do not provide a solution for outsourced k-NN because they requires data owners to store and compute on their data. Furthermore, they must remain online for all queries.



\begin{table}[t]
	\centering
  \begin{tabular}{ | l | c | c | c | }
  \hline
  Prior Work & DO & Q & Cloud \\  \hline
  Wong et al \cite{wong} & Trusted & Trusted & Semi-honest \\  \hline
  Zhu et al \cite{no-share-key} & Trusted & Semi-honest & Semi-honest \\  \hline
  Elmehdwi et al \cite{arxiv-knn} & Trusted & Semi-honest & Semi-honest \\  \hline
  Xiao et al \cite{knn-revisted} & Trusted & Trusted & Semi-honest \\  \hline
  \end{tabular}
  \vspace{2mm}
  \caption{A summary of the trust models from prior private outsource k-NN works. DO is the data owner, Q is the querier, and Cloud represents any cloud parties. Semi-honest parties are assumed to follow protocols but may attempt to discover query or data values.}
  \label{table:priorworks}  

\end{table}

The other line of prior privacy-preserving k-NN work relies on computing k-NN over encrypted data. These systems are designed for outsourcing, as depicted in Figure \ref{fig:systemold}. Table ~\ref{table:priorworks} summarizes the trust model each work assumes. An important observation is that all prior works assume a trusted data owner. 

Wong et al. \cite{wong} provides secure k-NN outsourcing by developing an
asymmetric scalar-product-preserving encryption (ASPE) scheme. ASPE
transforms data tuples and queries with secret matrices that are inverses of each
other. Multiplying encrypted tuples and queries cancel the transformations to
output scalar products, used for distance comparisons. Hence, a data owner and queriers can upload encrypted data tuples and queries to a cloud party, who can compute k-NN using scalar products. It is worth noting
this encryption scheme is deterministic, so identical data tuples have
the same ciphertext and likewise for queries.
The secret matrices are symmetric
keys for the encryption scheme, and must be shared with both the data owner and
queriers. This approach assumes a trusted data owner and queriers, and a semi-honest cloud party who follows the protocols but attempts to learn query or data values. 

In the outsourced k-NN system from \cite{no-share-key}, data encryption is again conducted by a trusted data owner, using a symmetric scheme with a  secret matrix transformation as a key. However, queriers 
do not share this key. Instead, they interact
with the data owner to derive a query encryption without revealing the query. Note this requires a data owner to always remain online for all queries. Also, data tuple encryption (being a matrix transformation)
is deterministic, but query encryption is not due to randomness introduced during the query encryption protocol. The encryption
scheme is designed similiarly to ASPE and preserves distance, so a cloud party can execute k-NN using distances computed from encrypted data tuples and queries. In this system's trust model, the data owner is trusted while the queriers and the cloud party are semi-honest.

The system in \cite{arxiv-knn} is designed to protect data and query privacy with two cloud parties. One cloud party is the data host, who stores all uploaded (encrypted) data tuples. The other cloud party is called the key party, since it generates the keys for a public-key encryption scheme. Data tuples and queries are encrypted with the key party's public key. In this system, Paillier encryption is used for its partially homomorphic properties. 
For each k-NN execution, the system computes encrypted distances through interactions between the key party and the data host. The key party orders the distances and provides the data host with the indices of the k nearest neighbor data tuples to return to a querier. This work assumes a trusted data owner, and semi-honest queriers and cloud parties.

Instead of finding exact k-NN, \cite{knn-revisted} allows a cloud party to approximate it
using encrypted data structures uploaded by the data owner that represent Voronoi boundaries. The Voronoi boundary surrounding a data point is the boundary equidistant from that data point to neighboring points. The region enclosed in a boundary is the region within which queries will return the same 1-nearest neighbor. Because the Voronoi boundaries change whenever a new data point is added, this approach prevents continuous data uploading
without redoing the entire outsourcing process. The encryption scheme is
symmetric, and both the data owner and queriers share the secret key. This work's trust model is identical to \cite{wong}'s, where the cloud party is semi-honest and the data owner and queriers are fully trusted.

Another contribution from \cite{knn-revisted} is a reduction-based impossibility proof for privacy preservation in outsourced exact k-NN under a single cloud party model, where the cloud party has access to an encryption key. Note that \cite{wong} and \cite{no-share-key} assume the cloud party does not have the encryption key, hence avoiding this impossibility argument.  The proof is a reduction to order-preserving encryption (OPE), and leverages prior work that shows OPE is impossible under certain cryptographic assumptions \cite{Boldyreva}. Fully homomorphic encryption does actually allow OPE but it is still impractical \cite{fhe}. Let $B$ be an algorithm, without access to a decryption oracle, that finds the nearest neighbor of an encrypted query $E(q)$ in the encrypted data $E(D)$. The impossibility proof shows that $B$ can be used to construct an OPE function $\mathcal{E}()$, hence $B$ cannot exist. However, their argument does not apply to a system model with multiple cloud parties.
Their proof, which relies on  $B$'s lack of access to a decryption oracle, can be circumvented by providing access to the decryption function at another cloud party, such as in \cite{arxiv-knn}. Furthermore, OPE has been realized in an interactive two-party setting \cite{stOPE}. As later discussed in Section \ref{sec:attacks} and \ref{sec:alternatives}, we consider it reasonable that a cloud party has encryption capabilities in a multi-data owner model. Thus further exploration of private k-NN is needed for scenarios not within the scope of this impossibility result.

One important observation about these prior works is that they all exhibit linear complexity in computation (and network bandwidth in the case of \cite{arxiv-knn}). Intuitively, this is because the protocols are distance based. Since the distances from a query to all data tuples varies for each query, all distances must be recomputed per query. While there are techniques \cite{kdtrees, localityhashing} for reducing the computational complexity of k-NN, these techniques may not be privacy preserving. For example, these algorithms necessarily compute on only a portion of the data tuples near the query. Observing similar data tuple subsets used can leak the similarity of subsequent queries. However, future work may yield more efficient privacy preserving k-NN constructions.

\section{System and Threat Models for Multiple Data Owners}

All existing private outsourced k-NN systems assume a single trusted data owner entity. In this paper, we are the first to consider multiple mutually distrusting data owner entities. This is a simple and practical extension of existing models, yet has important implications. We explore the various threat models that can arise in such a scenario, but we focus the most attention on one threat model we find particularly realistic. The remaining threat models can be more easily dealt with, for example by extending existing single data owner systems, and will be discussed in detail in Section \ref{sec:alternatives}. In this section, we first discuss our model of a multi-data owner system. We then provide a framework for modeling threats in the system.

\subsection{K-NN System Model}
\label{sec:systemmodel}

\begin{figure}
\centering
	\includegraphics[width=0.3\textwidth, height=75mm, keepaspectratio=false, angle=270]{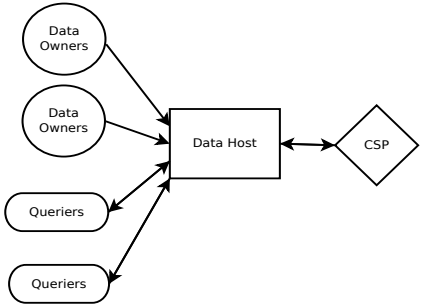}
	\caption{The model of an outsourced k-NN system containing four parties. The cloud parties are the data host and the crypto-service provider (CSP).}
	\label{fig:system}
\end{figure}

The privacy-preserving outsourced k-NN model has three types of parties: the data owners who outsource their data, the cloud party(s) who host the k-NN computation, and the queriers requesting k-NN classification. For emphasis, the key difference between our system model and prior models is the existence of multiple data owner entities. An immediate consequence of this modification is seen on the structure of the cloud parties. 
As discussed in Section \ref{sec:related}, results from \cite{knn-revisted} indicate that privacy-preserving outsourcing of exact k-NN cannot be achieved by a single cloud party without fully homomorphic encryption, which is impractical. Instead, at least two cloud parties should exist, one which stores encrypted data, and one with access to the associated decryption function that acts as a decryption oracle.
Hence our system model includes an additional cloud party we call the Cryptographic Service Provider (\csp) who generates a key pair for a public-key encryption scheme, and distributes the public key to data owners and queriers to use for encryption. The cloud party storing the encrypted data, termed the data host, is able to compute on encrypted data via interaction with the \csp. As depicted in Figure \ref{fig:system}, our system model involves multiple data owners and queriers, the data host, and the \csp.  Note that encrypted queries and data submissions must be sent to the data host, not the \csp, since the \csp\ can simply decrypt any received values.

\subsection{Threat Model Abstractions}
\label{sec:threatmodel}

We consider any party as potentially adversarial. Like prior works, we will consider semi-honest adversaries that aim to learn data or query values, possibly through collusion, while following protocols. We do not consider attacks where parties attempt to disrupt or taint the system outputs. Additionally, we must assume that the \csp\  and the data host do not directly or indirectly collude, since the \csp\ maintains the secret key to decrypt all data tuples on the data host. This is a reasonable assumption, for example, if the two cloud parties are hosted by different competing companies incentivized to protect their customers' data privacy.

To describe our threat models, we will take a slightly unorthodox approach. Instead of providing a model for each party's malicious actions, we model the malicious behavior of a party based on roles it possesses. The logic behind this approach is that different parties in our system model can pose the same threats because they possess the same roles. Enumerating and investigating all combinations of roles \textit{and} parties is redundant. Hence, just using roles provides a cleaner abstraction from which to analyze threats. Below we describe the four possible roles that arise in our system model.

\begin{itemize}

\item Data Owner Role: A party with this role may submit encrypted data into the system. We focus on misbehavior to compromise data privacy, and do not deal with spam or random submissions aimed at tainting system results. 
Note that the data owner role cannot compromise data privacy by itself since it does not allow observation of system outputs.

\item Querier Role: This role allows submission of encrypted queries to receive k-NN classification results.

We note that the querying ability can allow the discovery of the Voronoi cell surrounding a data point. The Voronoi cell is the boundary surrounding a point $p$ that is equidistant between $p$ and its neighbors. In the 1-NN case, queries within $p$'s Voronoi cell will return $p$'s classification. A query outside of the cell will return the classification of a neighboring point. Hence, changes in 1-NN outputs can signal a crossing over of a boundary. However, we deem this inherent leakage minimal since discovering the Voronoi cell would require numerous queries to reveal each boundary edge. The Voronoi cell also simply bounds the value of the data point, and does not directly reveal it. Furthermore, neighboring cells of the same class will appear merged in such an attack, since the output signal will not change between cells. K-NN with a large k parameter makes analysis even less accurate.

\item Data Host Role: The data host role can only be possessed by a cloud party. It allows storage and observation of incoming encrypted data tuples and queries, and the data host computation that is conducted. A holder of the role also observes any interactions with other parties.

\item \csp\ Role: Only a cloud party may possess this role. A \csp\ possesses both an encryption and decryption key, and can decrypt any data it observes. It interacts with a data host role to serve as a decryption oracle, and can observe any interaction with other parties, as well as the \csp
's own computation.

\end{itemize}

In Section \ref{sec:attacks}, we focus on the primary threat model in this paper, where any single party can possess both the data owner and querier role. In Section \ref{sec:alternatives}, the remaining threat models are explored. These are threats from one of the cloud parties possessing either the data owner role or the querier role, but not both. Also, we consider the case where the cloud parties possess neither roles.

\section{Attacks in the Data Owner-Querier Roles Threat Model}

\label{sec:attacks}

In this section, we consider the threat model where an adversarial party possesses both the data owner and querier roles (termed the DO-Q threat model). Since we consider all parties as potentially adversarial, any scenario where both roles may belong to a single party falls under this threat model. We argue this is a realistic threat model, and discover a set of conceptually simple attacks on any system regardless of system design or encryption scheme. The attacks work based purely on the mathematical nature of the k-NN output, rather than implementation specifics. Hence, we conclude that multi-data owner outsourced k-NN cannot be secured under such a threat model.

\subsection{DO-Q Threat Model}

The DO-Q threat model covers any situation where a single party can possess both the data owner and querier roles. We consider this a practical threat model because it can arise in numerous scenarios. Hence, we will focus on it in this section as well as Section \ref{sec:kernelregression}. The following outlines the possible scenarios where a single party may possess both roles:

\begin{itemize}

\item It is reasonable to expect that data owners, contributing their sensitive data, are  allowed access to querying. If not, there is less incentive for data owners to share data. For examples, hospitals might want to pool their data to create a disease classification system. The doctors at any participating hospital should be able to query the system when diagnosing their patients.

\item The data owners may not be explicitly given querying permissions, say if the data owners and queriers are separate institutions. However, miscreants in each party may collude together, providing a joint alliance with both roles.

\item If the data host can encrypt data tuples and queries, it can act as its own data owner and querier. It can insert its own encrypted data tuples into the system's data set, and compute k-NN using its own encrypted queries. A data host with access to the encryption keys is not unreasonable. In \cite{arxiv-knn}, the data host needs to encrypt values to carry out operations on encrypted data tuples. In addition, the queries are encrypted under the same key as data tuples to allow homomorphic operations.

\item If the data host lacked the data owner and/or querier roles (e.g., if an encryption key was kept secret from it), it may still collude with a data owner and/or querier to obtain the roles. This collusion can also occur between the \csp\ with a data owner and querier.

\item If the system is public to data submissions, queries, or both, then any party can supplement their current roles with the public roles. For example, a fully public system allows anyone to obtain both the data owner and querier role.

\end{itemize}

\subsection{Distance-Learning Attacks in the DO-Q Model}

We now present a set of attacks that reveal distances between a query and encrypted data tuples, allowing triangulation of the plaintext data values. We begin by presenting attacks under the  simpler 1-NN, and gradually progress towards k-NN. We assume that k-NN does not output the neighboring tuples directly, which provides arbitrary privacy breaches through querying, but rather just the query's predicted classification.

We note that prior work \cite{esa_attacks} has developed algorithms for cloning the Voronoi diagram of 1-NN using only queries if the query responses contain the exact location of the nearest neighbor or the distance and label of the nearest neighbor. 
If only the nearest neighbor label is returned, then the algorithm provides an approximate cloning.
Our attacks differ in that they are structurally very different, we look at k-NN beyond 1-NN, our attacks reveal the exact data value of a target tuple rather than the Voronoi diagram, and our attacks leverage data insertion as well as querying. Also we focus on a system model where query responses do not contain the distance (which we considered an already broken construction). The algorithms in \cite{esa_attacks} require at least distance in the query response to conduct exact cloning.

\subsubsection{Attack using Plaintext Distance}

We begin by considering a broken construction of outsourced k-NN. The k-NN system must calculate the distances from all data tuples to the query point. If these distances are ever visible in plaintext to a party with just the querier role, then it is simple to discover the true value of any encrypted data tuple.

Knowing a query $q$, the adversary can observe the distance $l$ from $q$ to the nearest neighbor. This forms a hypersphere of radius $l$ around $q$. If the data is $d$ dimensions, $d+1$ different query points with the same nearest neighbor constructs $d+1$ hyperspheres, which will intersect at one unique point, the nearest neighbor's plaintext value. Hence, we can uncover data even though it is encrypted in the data set. Note that in this case, the adversary needs both a querier role as well as the cloud party role that observes plaintext distances. This is different from the threat model we are considering, but we discuss it as it provides insight for the following attacks.

\subsubsection{Attack on 1-NN}

Now consider the 1-NN scenario where the system does not foolishly reveal the distance in plaintext. An adversary with the data owner and querier roles can still breach the privacy of the system. The querier role allows the adversary to submit a query $q$, and observe the classification $C$ of its nearest neighbor $p$. The attacker, using its data owner role, can then insert an encrypted ``guess'' data tuple $E(g)$, with any classification $C' \neq C$. If the new nearest neighbor classification returns $C'$, we know $p$ is farther away from $q$ than $g$. If not, then $p$ is closer. Using this changing 1-NN output signal, the adversary can conduct binary search using additional guess tuples to discover the distance from $q$ to $p$. Hence, the distance is revealed and triangulation can be conducted as in the insecure previous case. This takes $O((d+1) \log D)$ guesses in total to compromise a data tuple, where $d$ is the data dimensionality and $D$ is the distance from $q$ to the initial guess insertion. 

Note the above attack appears to require tuple deletions or updates, without which guess tuples that are closer than $p$ will disallow continued binary search. Deletions or updates will certainly improve attack performance, and is reasonable to allow in many scenarios (e.g., location data which is constantly changing). However, it is not required. First, the attacker could conduct a linear search, rather than binary, starting with a guess at distance $D$ and linearly decreasing the guess distance until equal to the distance between $q$ and $p$. Alternatively, a too-near guess still narrows down the range in which $p$'s value can be located, and the adversary can restart a search with a new query in that range.

\subsubsection{Attack on k-NN}

Consider now attacks on k-NN, instead of 1-NN. In this scenario, the k-NN system returns the classes of all $k$ nearest neighbors. Let $p$ of class $C$ be the nearest neighbor to query $q$. In fact, k-NN can be reduced to 1-NN by inserting $k-1$ data tuples of class $C' \neq C$ adjacent to $q$. Then, $p$ is the only unknown data tuple whose classification is included in the k-NN output set. The attacker with data owner and querier roles can test guess data insertions of class $C'$ until $C$ is no longer included in the k-NN output, indicating the guess tuple is closer than $p$. This signal can be used as before to conduct a distance search.

\subsubsection{Attack on Majority-Rule k-NN}

The final scenario for a k-NN system is one that operates on majority rule. Rather than outputting all classifications from the $k$ nearest neighbors, only the most common classification is returned.
Again, data owner and querier privileges allow a distance-learning attack by reducing to 1-NN. Let $p$ of class $C$ be the nearest neighbor. The attack can insert $k-1$ data tuples adjacent to a query $q$, split evenly over all classes such that the output of k-NN solely depends on the classification of $p$. For example, if we have binary classification (0 and 1) and 3-NN, the attacker can insert a 0-class tuple and a 1-class tuple adjacent to $q$. The original nearest neighbor ($p$) now solely determines the output of 3-NN, and the adversary may again use a change in k-NN output as a signal for a distance search.

\subsection{Fundamental Vulnerability}

The vulnerability exposed by these attacks is fundamental to exact k-NN. For any given query, k-NN outputs a function evaluated on the subset of tuples within a distance $d$, where $d$ is the distance of the $k$-th nearest neighbor. Representing k-NN as kernel density estimation using a uniform kernel, $d$ is the kernel's width. The vulnerability above fundamentally relies on the fact that the kernel's width changes depending on which data points are near the query.
An adversary can learn distances through guess insertions and an observation of k-NN output change. Hence, in any scenario where a party possesses both the data owner and querier roles, exact k-NN cannot be secure, regardless of implementation details. An alternate or approximate solution must be used. In Section \ref{sec:kernelregression}, we propose a privacy-preserving scheme using a similar algorithm from the same family as k-NN.

\section{Privacy-Preserving Kernel Density Estimation}
\label{sec:kernelregression}

In Section \ref{sec:attacks}, we demonstrated that data privacy can be breached in any multi-data owner exact k-NN system under the DO-Q threat model. Given the practicality of this scenario, particularly since the data host typically will have those roles, we seek to provide a secure alternative to exact k-NN. In this section, we propose a privacy-preserving system using an algorithm from the same family as k-NN.

In particular, k-NN is a specific case of kernel density estimation (as discussed in Section \ref{sec:background-kde}). Intuitively, the kernel measures the amount of influence a neighbor should have on the query's final classification. For a given query, kernel density estimation computes the sum of the kernel values for all data points of each class. The classification is the class with the highest sum. K-NN uses a kernel that is uniform for the $k$ nearest neighbors, and zero otherwise. We propose substituting this kernel with another common kernel, the Gaussian kernel:
\begin{align*}
  K(\|\mathbf{q} - \mathbf{x}_i\|) &=  \frac{1}{\sigma\sqrt{2\pi}}e^{-\frac{1}{2\sigma^2}\|\mathbf{q} - \mathbf{x}_i\|^2}
\end{align*}

In this equation, $\mathbf{q}$ is the query tuple, $\mathbf{x}_i$ is the $i$-th data tuple, $\sigma$ is a parameter (chosen to maximize accuracy using standard hyperparameter optimization techniques, such as cross-validation with grid search), and $\| \cdot \|$ is the L2 norm. As the distance between $q$ and $x_i$ increases, the kernel value (and influence) of $x_i$ decreases.
K-NN allows the $k$ nearest neighbors to all provide equal influence on the query point's classification. When using a Gaussian kernel, all points have influence, but points farther away will influence less than points nearby. While this approach is \textit{not} equivalent to exact k-NN due to the non-uniformity of neighbor influences, it is similar and from the same family of algorithms. Later, we show experimentally that it can serve as an appropriate substitute in order to provide privacy preservation.

\subsection{Why the Gaussian Kernel?}
\label{sec:intuition}

The Gaussian kernel for kernel density estimation is specifically chosen to provide defense against the adaptive k-NN distance-learning attacks under the DO-Q threat model (Section \ref{sec:attacks}). Here, we briefly provide the intuition for why this is true, with the formal proof in Section \ref{sec:secanalysis}. 

In k-NN, the width of the kernel depends on the distance of the $k$th nearest point from the query, so that only the $k$ nearest neighbors influence the classification. An adversary with the data owner role can change this set of neighbors by inserting data tuples within the kernel width, and manipulate whether a particular point influences the classification.  Distances can be learned by using changes in this subset as a signal. When using the Gaussian kernel, which has an infinite kernel width, any data the adversary inserts will not change the influence of other points, and the outcome will be affected by exactly an amount the adversary already can compute. For example, for a given query $q$, an adversary inserting a point $y$ knows the influence for $y$'s class will increase by exactly $K(\|\mathbf{q} - \mathbf{y}\|)$. Hence, the adversary learns nothing new from possessing the data owner role and querying. 

It is important to realize that this security arises from the Gaussian kernel's unvarying and infinite width. 
The choice of an alternative kernel is non-trivial and must be carefully selected. For example, another viable substitute is the logistic kernel $K(u) = \frac{1}{e^u+2+e^{-u}}$, since it too has an unvarying and infinite width. Our decision to use the Gaussian kernel was based on the ease of developing a scheme using partially homomorphic cryptographic primitives.

\subsection{A Privacy-Preserving Design}
\label{sec:protocol_design}

In this section, we will step-by-step describe the construction of a privacy-preserving classification system using kernel density estimation with a Gaussian kernel. As we will demonstrate, our protocols provide classification without leaking data or queries. Our system will follow the same system model as described in Section \ref{sec:systemmodel}. 

\subsubsection{Setup}
Recall that the system model provides outsourcing data storage and computation to a data host and a cryptographic service provider. The cryptographic service provider generates a public/private key pair for the Paillier cryptosystem and distributes the public key to data owners and queriers, who use it to encrypt data they submit to the cloud, as well as the data host. Data owners submit their data in encrypted form to the data host. Note that while we use the Paillier cryptosystem, other similarly additive homormophic cryptosystems may be appropriate as well.

\subsubsection{Computing Squared Distances}

Since a kernel is a function of distance, our system must be able to compute distance using encrypted data without revealing it. Algorithm \ref{alg:sqdist} describes $SquaredDist(E(a),E(b))$, which allows the data host to compute the encrypted squared distance between $a$ and $b$ given only their ciphertexts. The protocol does not leak the distance to either cloud parties, to prevent a distance-learning attack. Only squared distances are required as they are used in the Gaussian kernel calculation.

Assume our tuples are $m$ dimensions, and $a_i$ is the $i$-th feature of tuple $a$. $E$ and $E^{-1}$ are the Paillier encryption and decryption functions, respectively, using the previously-chosen public/private key pair. Note all Paillier operations are done modulo the Paillier parameter $N$, but for simplicity we elide the modulus.  Also, we abbreviate the data host as DH, and the cryptographic service provider as \csp.

\begin{algorithm}[t]
\caption{$SquaredDist(E(a), E(b))$: Output $E(\|a-b\|^2)$}
\begin{algorithmic}
\item[1.] $DH$: \textbf{for} $1 \leq i \leq m$ \textbf{do:}
\item[\hspace{3mm} (a)] $x_i \leftarrow E(a_i) \cdot E(b_i)^{-1}$. Thus, $x_i = E(a_i - b_i)$.
\item[\hspace{3mm} (b)] Choose random $\mu_i \in \mathbb{Z}_N$. 
\item[\hspace{3mm} (c)] $y_i \leftarrow x_i \cdot E(\mu_i)$, such that $y_i = E(a_i - b_i + \mu_i)$.
\item[\hspace{3mm} (d)] Send $y_i$ to $CSP$.
\item[2.] $CSP$: \textbf{for} $1 \leq i \leq m$ \textbf{do:}
\item[\hspace{3mm} (a)] $w_i \leftarrow E((E^{-1}(y_i))^2)$.
\item[\hspace{3mm} (b)] Send $w_i$ to $DH$.
\item[3.] $DH$: \textbf{for} $1 \leq i \leq m$ \textbf{do:}
\item[\hspace{3mm} (a)] $z_i \leftarrow w_i \cdot x_i^{-2\mu_i} \cdot E(-\mu_i^2)$.
\item[4.] $DH$: Encrypted squared distance $E(\|a-b\|^2) = \Pi_{i=1}^m z_i$.
\end{algorithmic}
\label{alg:sqdist}
\end{algorithm} 

Step 1 computes the encryption of $d_i = (a_i-b_i)$ and additively masks each with a random secret $\mu_i$, such that $E^{-1}(y_i)=d_i+\mu_i$ $mod$ $N$. This prevents the \csp\ from learning $d_i$ in step 2. The \csp\ decrypts $y_i$, squares it, and sends it back to the DH the encryption of the square. Note $w_i = E((E^{-1}(y_i))^2) = E(d_i^2 + 2 \mu_i d_i + \mu_i^2)$. The DH can encrypt $\mu_i^2$ and compute $E(2\mu_i d_i)$ as $x_i^{2\mu_i}$, allowing it to recover $E(d_i^2)$ in step 3. Finally, the DH can compute $E(d^2) = E(\Sigma_{i=1}^m d_i^2) = \Pi_{i=1}^m E(d_i^2)$, as in step 4. This provides the DH with encrypted squared distances without revealing any information to either cloud parties.

\subsubsection{Computing Kernel Values}

Now that we can compute squared distances, we must compute the Gaussian kernel values in a secure fashion. Algorithm \ref{alg:kernel} does so such that the \csp\ obtains a masked kernel value for each data tuple.

\begin{algorithm}[h!]
\caption{$KernelValue(q)$: Compute Gaussian kernel values for data tuples $t$ in data set $D$ given query $q$}
\begin{algorithmic}
\item[1.] $DH$: \textbf{for} $1 \leq i \leq |D|$ \textbf{do:}
\item[\hspace{3mm} (a)] $s_i \leftarrow SquareDist(t_i, q) = E(\|t_i - q\|^2)$
\item[\hspace{3mm} (b)] Choose random $\mu_i \in \mathbb{Z}_N$
\item[\hspace{3mm} (c)] $e_i \leftarrow s_i \cdot E(\mu_i)$. Thus, $e_i = E(\|t_i - q\|^2 + \mu_i)$.
\item[\hspace{3mm} (d)] Send $e_i$ to $CSP$.
\item[2.] $CSP$: \textbf{for} $1 \leq i \leq |D|$ \textbf{do:}
\item[\hspace{3mm} (a)] $g_i \leftarrow \frac{1}{\sigma\sqrt{2\pi}}e^{-\frac{1}{2\sigma^2} E^{-1}(e_i)}$
\end{algorithmic}
\label{alg:kernel}
\end{algorithm} 

In step 1, the DH additively masks each computed encrypted squared distance by adding a different random mask. These values are sent to the \csp\, who in step 2 decrypts and computes the kernel values using the masked squared distances. Note that since each distance is masked  with a random value, the \csp\ does not learn any distances, even if the \csp\ knows some data values in the data set. The \csp\ will not be able to determine which masked values correspond to its known values. Again, no cloud party learns any distances, and they observed only encrypted or masked values. The query is only used in this algorithm, without ever being decrypted. Hence query privacy is achieved.
\newpage

\subsubsection{Computing Classification}

The final step is to conduct kernel density estimation for classification by determining which class is associated with the largest sum of kernel values. Our scheme represents classes not as a single value, but rather as a vector. If there are $c$ classes, we require a set of $c$ orthonormal unit-length bases, where each basis is associated with a particular class. When uploading a data tuple, a data owner uploads the encrypted basis associated with the tuple's class.

We can compute classification as described in Algorithm \ref{alg:class}. This algorithm is run after Algorithm \ref{alg:kernel}, such that the \csp\ knows the masked Gaussian kernel values for all data tuples.  Let $c_i$ represent the classification vector for the $i$-th data tuple.

\begin{algorithm}[ht!]
\caption{$Classify(q)$: Output a classification prediction for query $q$ where there are $c$ classes}
\begin{algorithmic}
\item[0.] $DH$: Run $KernelValue(q)$
\item[1.] $DH$: \textbf{for} $1 \leq i \leq |D|$ \textbf{do:}
\item[\hspace{3mm} (a)] Generate a random invertible matrix $B_i$ of size $c \times c$.
\item[\hspace{3mm} (b)] $v_i = B_i \times E(c_i)$.
\item[\hspace{3mm} (d)] Send $v_i$ to $CSP$.
\item[2.] $CSP$: \textbf{for} $1 \leq i \leq |D|$ \textbf{do:}
\item[\hspace{3mm} (a)] $w_i = v_i^{g_i}$ , where $g_i$ are kernel density values.
\item[\hspace{3mm} (b)] Send $w_i$ to $DH$.
\item[3.] $DH$: \textbf{for} $1 \leq i \leq |D|$ \textbf{do:}
\item[\hspace{3mm} (a)] $w'_i \leftarrow B_i^{-1} \cdot w_i^{e^{\frac{1}{2\sigma^2}\mu_i}}$, where $\mu_i$ is the random distance masking for the $i$-th tuple (see Algorithm \ref{alg:kernel} step 1(c))
\item[4.] $DH$: $A \leftarrow \Sigma_{i=1}^{|D|} w'_i$. Note that this is a $c \times 1$ vector.
\item[5.] $DH$: \textbf{for} $1 \leq i \leq c$:
\item[\hspace{3mm} (a)] Choose new $\mu_i \in \mathbb{Z}_N$ randomly.
\item[\hspace{3mm} (b)] $A_c \leftarrow A_c + \mu_i$, where $A_c$ is the $c$-th component of $A$.
\item[6.] $DH$: Let $f(in_1,...,in_c, \mu_1,...,\mu_c) $\\$\hspace{15mm}= \underset{k}{\arg\max} (in_k - \mu_k \mod N)$. 
\item[\hspace{3mm} (a)] Generate the garbled circuit $GF_f$ of $f$ and the garbled inputs $GI(\mu_i )$ $\forall i \in \{1,...,c\}$. 
\item[\hspace{3mm} (b)] Send $GF_f$, $GI(\mu_i)$ and $A_i$ $\forall i \in \{1,...,c\}$ to $CSP$.
\item[7.] $CSP$: 
\item[\hspace{3mm} (a)] $A'_i \leftarrow E^{-1}(A_i)$ $\forall i \in \{1,...,c\}$.
\item[\hspace{3mm} (b)] Conduct oblivious transfers with $DH$ to obtain $GI(A'_i)$ $\forall i \in \{1,...,c\}$
\item[\hspace{3mm} (c)] $gout \leftarrow GF_f(GI(A'_1),...,GI(A'_c),GI(\mu_1),...,GI(\mu_c))$.
\item[\hspace{3mm} (d)] Send $gout$ to $DH$.
\item[8.] $DH$: Map $gout$ to its non-garbled value $out$ and return to the querier. $out$ is the number representing the predicted class.
\end{algorithmic}
\label{alg:class}
\end{algorithm}

In step 1, the encrypted classification vectors are masked through matrix multiplication with a random invertible matrix $B_i$, which can be efficiently constructed \cite{random_matrix}. Note this computation can be conducted on encrypted data because the DH knows the values of $B_i$, and the matrix multiplication involves only multiplying by plaintext constants and adding ciphertexts. The masked classification vectors are sent to the \csp\ in step 2, who scales them by the kernel values (from Algorithm \ref{alg:kernel}). In step 3, the DH is returned these values. The DH undoes the transformation by multiplying by $B_i^{-1}$. The distance masks in the Gaussian kernel from Algorithm \ref{alg:kernel} are removed by multiplying by $e^{\frac{1}{2\sigma^2}\mu_i}$, where $\mu_i$ is the mask for the $i$-th tuple. Now, $w'_i$ is the encrypted basis for the $i$-th tuple's class, scaled by the kernel value. In step 4, the sum of these scaled encrypted bases will sum the kernel values for each class in the direction of the classes' bases, forming the vector $A$. Note that $A$ is still a vector of encrypted values though.

At this point, we need to determine the class with the highest kernel values summation. Having the \csp\ decrypt $A$ and return the kernel value sums for the classes may appear straightforward, but in fact leads to an attack on future inserted data tuples if an adversary can view the kernel value sums and can continuously query $q$. When the next data tuple $t$ is inserted, the adversary will observe an increase in the kernel values summation for $t$'s class by $K(\|t-q\|^2)$. Knowing $q$ and $K(\|t-q\|^2)$, the adversary can compute $\|t-q\|$. Assuming the data has $m$ features, the attacker can determine $t$ using a set of $(m+1)$ queries both before and after insertion to compute $(m+1)$ distances. Note that our threat model allows for one of the cloud parties to possess the querying role, so whichever party can observe the kernel value sums can conduct this attack. Thus, the protocol must determine the classification without revealing the kernel value sums to any party. 

We accomplish this using a garbled circuit for the function $f(in_1,...,in_c, \mu_1,...,\mu_c) = \underset{k}{\arg\max} (in_k - \mu_k \mod N)$. If there are $c$ classes, this function accepts $c$ masked inputs $in$ and $c$ additive masks $\mu$, unmasks each masked input with its associated mask, and returns the index corresponding to the largest unmasked value. If the masked inputs are masked kernel value summations for each class, $f$ returns the class index with the kernel values summation. In step 5, the DH additively masks the (encrypted) kernel values summation for each class. Then in step 6, the DH generates the garbled circuit $GF_f$ for $f$ and garbles its own inputs, which are the masks for each class's kernel values summation. Note that $f$ is not that complex of a function (for example, compared to a decryption function) and its garbled circuit can be practical implemented using existing systems \cite{yao_implementation}. 
The DH sends $GF_f$, its garbled inputs, and the masked (still encrypted) $A$ to the \csp. The \csp\ decrypts the masked $A$ vector in step 7(a). By conducting oblivious transfer in step 7(b) with the DH to obtain the garbled input associated with each $A_i$, the $CSP$ does not reveal the masked kernel value sums. With all of the inputs now to $GF_f$, the \csp\ evaluates the garbled circuit and returns the garbled output to the DH. Because the DH created the circuit, it can map the garbled output back to its non-garbled value, which is the class index returned by $f$ with the largest kernel values summation. This classification is exactly the classification output of kernel density estimation using a Gaussian kernel.

\subsection{Security Analysis}
\label{sec:secanalysis}

First we must show that use of the Gaussian kernel by our scheme is resistant to the adaptive distance-learning attacks detailed in Section \ref{sec:attacks}. While Section \ref{sec:intuition} provide an intuitive argument for this, here we present a formal proof.

\begin{theorem}
\label{proof:nothing}
Under the Gaussian kernel in kernel density estimation, the adversary learns nothing from the distance-learning attacks of Section \ref{sec:attacks}. More specifically, the adversary does not learn the distance from a query output by adding to the data set (or removing her own data tuple).
\end{theorem}

\begin{proof}
    Let $D=\{d_0, ..., d_n\}$ be the current data set, $q$ be the adversary's desired query point, and $d_A$ be any data tuple the adversary could insert. Also define $GK_j(D, q)$ to be the kernel value sums for the $j$-th class given the query $q$ and data set $D$, and $C(d_i)$ to be the $i$-th data tuple's classification. Recall that if the system allows for $c$ classes, the query output is $\underset{j \in \{1,...,c\}}{\arg\max} (GK_j(D,q))$.

    Before tuple insertion, for each class $j$ in $\{1,...,c\}$: 
\begin{align*} 
    GK_j(D, q) = \Sigma_{i=1}^{|D|} \frac{1}{\sigma\sqrt{2\pi}}e^{-\frac{1}{2\sigma^2} (\|d_i - q\|^2)} \mathbbm{1}_{C(d_i)=j}
\end{align*}
Here, $\mathbbm{1}_{C(d_i)=j}$ is an indicator variable for whether the $i$-th data tuple $d_i$ is of class $j$. 
After inserting tuple $d_A$, for each class $j$ in $\{1,...,c\}$:
\begin{align*} 
    GK_j(D \cup \{d_A\}, q) = \Sigma_{i=1}^{|D|} \frac{1}{\sigma\sqrt{2\pi}}e^{-\frac{1}{2\sigma^2} (\|d_i - q\|^2)} \mathbbm{1}_{C(d_i)=j} \\ + \frac{1}{\sigma\sqrt{2\pi}}e^{-\frac{1}{2\sigma^2} (\|d_A - q\|^2)} \mathbbm{1}_{C(d_A)=j}\\ = GK_j(D, q) + \frac{1}{\sigma\sqrt{2\pi}}e^{-\frac{1}{2\sigma^2} (\|d_A - q\|^2)} \mathbbm{1}_{C(d_A)=j}
\end{align*}

The adversary, through using its data owner role, can cause a change of $\delta = \frac{1}{\sigma\sqrt{2\pi}}e^{-\frac{1}{2\sigma^2} (\|d_A - q\|^2)} \mathbbm{1}_{C(d_A)=j}$. 
However, the Gaussian kernel values have not changed for any other data tuples. 
The change in a query output is directly and completely a result of $\delta$, which the adversary knows independent of querying since it knows $d_A$, $q$, and $\sigma$ (which can be public).
Given the query output is based only on hidden Gaussian kernel value sums, the output after insertion does not leak any more information on an individual tuple's kernel value beyond what the query output leaks inherently (without insertion).
Because the distance is only used in the Gaussian kernel, and the adversary learns nothing about other tuples' Gaussian kernel values through insertion, it does not learn the distance from insertion.

If the adversary deleted one of her tuples, the outcome is exactly the same except $\delta$ is negative. With a semi-honest adversary, we can assume the adversary only deletes her own tuples.
\end{proof}

The next analysis is to show our scheme leaks no further data beyond its output. More formally, our scheme leaks no additional information compared to an ideal scheme where the data is submitted to a trusted third party, who computes the same classification output. In our initial discussion of the protocol design in Section \ref{sec:protocol_design}, we justified the security along this front for each step of the construction. See Appendix \ref{sec:additional_security_analysis} \fullpaperappendix for a more formal proof.

A final analysis that could be conducted is to formally characterize what the kernel density estimation output itself reveals about the underlying data. One way is to use the framework of differential privacy~\cite{dwork2006differential}, which can tell us how much randomness we need to add to the output in order to achieve a specific quantifiable level of privacy as defined by the framework.

We say that a randomized algorithm $A$ provides $\epsilon$-differential privacy if for all data sets $D_1$ and $D_2$ differing by at most one element, and all subsets $S$ of the range of $f$, \[\Pr(A(D_1) \in S) \le \exp(\epsilon) \cdot \Pr(A(D_2) \in S).\]
We also define the \emph{sensitivity} of a function $f: \mathcal{D} \rightarrow \mathbb{R}^m$ as \[\Delta f = \max_{D_1, D_2} \|f(D_1) - f(D_2)\|_1.\]
Then \cite{dwork2006differential} shows that, if we add Laplacian noise with standard deviation $\lambda$, as in
\[\Pr(A_f(D) = k) \propto \exp(-\|f(D) - k\|/\lambda),\] $A_f$ gives $(\Delta f / \lambda)$-differential privacy.

In our case, we can view $f$ as performing a kernel density estimation query with a particular query point on a set of data points, to receive a score for each class.
$A_f$ would be a randomized version of $f$, which takes these scores from $f$ and adds Laplacian noise to each one.
The sensitivity is $\Delta f = \frac{1}{\sigma \sqrt{2\pi}}$, as follows from the proof of Theorem~\ref{proof:nothing} when $\|d_A - q\| = 0$, and so using $A_f$ would grant us $\frac{\sigma \sqrt{2\pi}}{\lambda}$-differential privacy.
We can see that a larger $\sigma$ in the Gaussian kernel grants us greater differential privacy (i.e., leads to a higher value of $\epsilon$) when the magnitude of the noise remains fixed.
This follows the intuition that a larger $\sigma$ in the kernel increases the influence of far-away points in the data set, and so the result is not as effected by any particular point.

While we proved that the Gaussian kernel output does not leak data tuple values through the distance-learning attacks of Section \ref{sec:attacks}, we have not yet developed a framework to analyze other leakage channels other than the preliminary analysis above. Such a problem is challenging and we leave it to future work. However, even in the face of potential leakage, data perturbation provides some security, although in exchange for output accuracy. 

\subsection{Performance Considerations}

Our kernel density solution requires $O(N)$ computation and communication overhead, where $N$ is the number of data tuples in the data set. On large data sets, we acknowledge this cost might be impractical. Future work remains to find privacy-preserving constructions with optimizations. Our goals with the presented construction are to both demonstrate that there exists a theoretically viable privacy-preserving alternative to k-NN, and lay the groundworks for future improvements. Note that existing private k-NN schemes also have linear performance costs, as discussed in Section \ref{sec:related}.

One solution for improving performance is parallelism, since our scheme operates on each data tuple independent of other tuples. Ideally, multiple machines would run in parallel at both cloud parties, providing parallelism not only in computation but also network communication. Each machine at the data host would compute a subcomponent of the kernel value sums, and one master node can aggregate these subcomponents and execute Yao's garbled circuit protocol to produce the classification. Our protocol was intentionally designed so that the function circuit that is garbled is of limited complexity, without need for complex sub-routines such as decryption. Thus, Yao's protocol can be implemented efficiently \cite{yao_implementation}.

K-NN suffers this linear-growth complexity as well, however approximation algorithms significantly improve performance by reducing the number of points considered during neighbor search. For example, k-d trees \cite{kdtrees} divide the search space into smaller regions, while locality hashing \cite{localityhashing} hashes nearby points to the same bin. Hence, search for a query's neighbors involves computation over the data points in the same region or bin, rather than the entire data set.
Unfortunately, it is not obvious how to execute approximate k-NN in a privacy-preserving fashion. 
Similarly, there must be future work on similar optimizations and approximations for kernel density estimation.

\section{Comparing k-NN and Gaussian Kernel Density Classification}

\begin{table}[b!]
  \centering
  \begin{tabular}{c|ccc}
    \textbf{Data Set} & \textbf{k-NN} & \textbf{KDE} & \textbf{Agreement} \\ \hline
    Cancer 1 & 97.85\% & 97.14\% & 99.28\% \\
    Cancer 2 & 96.49\% & 96.49\% & 98.24\%\\
    Diabetes & 81.82\% & 74.68\% & 87.66\% \\
    MNIST & 97\% & 96\% & 99\%
  \end{tabular}
  \vspace{2mm}
  \caption{Empirical results for k-NN versus kernel density estimation (KDE) with a Gaussian kernel. Using several data sets, we evaluate each algorithm's classification accuracy as well as the degree of agreement between the two algorithms' outputs.}
  \label{table:knn-vs-kde}
\end{table}

In this section, we examine the differences between classification with k-NN and with Gaussian kernel density estimation. We empirically evaluated the performance of using both algorithms, to show they can give similar results on real data. Appendix \ref{sec:practical_differences} \fullpaperappendix expands on this evaluation, discussing a number of additional considerations when using Gaussian kernel density estimation. We construct hypothetical data sets where we can expect them to behave differently, and discuss the underlying causes for the discrepancy.
We also discuss some practical issues specific to using a Gaussian kernel and how to solve them.

We summarize the results in Table~\ref{table:knn-vs-kde}. The medical data comes from the UCI Machine Learning Repository~\cite{Bache+Lichman:2013}, a collection of real data sets used for empirical evaluations of machine learning algorithms.
``Cancer 1'' and ``Cancer 2'' contain measurements from breast cancer patients conducted at the University of Wisconsin. ``Cancer 1'' contains 9 features for each patient; ``Cancer 2'' came several years after ``Cancer 1'' and contains 30 features which are more fine-grained.
``Diabetes'' contains measurements from Pima Indian diabetes patients (8 features per patient). 
We also use the MNIST handwritten digit recognition data set, which consist of $28 \times 28$ pixel grayscale images each containing a single digit between $0$ and $9$. 

Following our recommendations described in Appendix \ref{sec:practical_differences}, we scaled each feature to $[0, 1]$. For the UCI data sets, we used cross-validation to select $k$ and $\sigma$; for MNIST, we set them manually to $5$ and $0.25$.
Except for MNIST which already came divided into training and test sets, we randomly took 20\% of each data set for testing. Table~\ref{table:knn-vs-kde} contains the prediction accuracy and classification agreement for both algorithms on these test sets when using the remaining 80\% as training data.

Except for the Pima Indian diabetes data set, Gaussian kernel density classification exhibits an accuracy no more than $1\%$ lower than k-NN's, and the two algorithms agree on over $99\%$ of classifications. In the diabetes case, both algorithms perform poorly because the data is not well separated, with regions scattered with both classes. We argue such a data set is not well suited for k-NN in the first place. Prior investigation must decide whether a particular algorithm is suitable for use.

\section{Alternative Threat Models}
\label{sec:alternatives}

Sections \ref{sec:attacks} and \ref{sec:kernelregression} focused on the DO-Q threat model since it accounts for a number of realistic scenarios, including various collusion relationships. Without any party (directly or through collusion) possessing both the data owner and querier roles, there are several remaining threat situations. In this section, we explore these threat models. The one pervasive assumption remains: the data host and the cryptographic service provider (\csp) do not collude. As mentioned in Section \ref{sec:threatmodel}, a data owner or querier role itself does not allow privacy compromise of the system. Hence, the remaining alternate threat models are outlined below:
\begin{itemize}

\item Each party possesses only its original role (e.g., a data host only possesses a data host role).
\item The data host additionally obtains a data owner role.
\item The data host additionally obtains a querier role.
\item The \csp\ additionally obtains a data owner role.
\item The \csp\ additionally obtains a querier role.

\end{itemize}
Note that a given system may be under more than one of the above threats. For example, the \csp\ and data host may both obtain an additional role. However, since we assume no collusion directly or indirectly between the two cloud parties, we can analyze each threat independently.

These threat models, while each covering fewer scenarios, are still realistic. There are settings where the cloud parties do not collude with certain other parties. For example, the cloud parties may have reputations to maintain, and may avoid collusion or misbehavior for fear of legal, financial, and reputational backlash from discovery through security audits, internal and external monitoring, whistleblowers, or investigations. Note that in all cases, our kernel density estimation solution using the Gaussian kernel can still be used as an appropriate alternative, since the DO-Q threat model allows for a strictly more powerful (in terms of roles) adversary. Our discussion that follows will look at other solutions though, under each threat model.

\subsection{Original Roles Threat Model}

In this threat model, each party only possesses or utilizes the role it was originally designated. This is the simplest of threat models because data owners do not collude with any other parties. Since data owners themselves do not receive output from the k-NN system, they cannot compromise privacy. Hence, we can treat them as trusted, allowing us to revert to single data owner designs \cite{arxiv-knn, wong}, except using multiple data owners. These existing systems already provide privacy against cloud and querier adversaries. 
This scenario could realistically arise if data owners and querier are separate institutions. For example, hospitals could contribute data to a research study conducted at a university, and only the researchers at that university can query the system.

\subsection{Data Host Threat Models}

We consider now the threat models where the data host obtains the role of a data owner or a querier, but not both. This can be through collusion, or a design giving the data host the ability to insert data (e.g., encrypt data tuples) or initiate queries. 

An extension of the system in \cite{arxiv-knn} can protect against a data host with the data owner role. This system uses computation over Paillier ciphertext, similar in flavor to our kernel density estimation approach. Their system consists also of a \csp\, which they call the key party, and a data host. 
Their scheme is privacy-preserving in the single data owner model. The data host is allowed to encrypt, and hence possesses the role of a data owner, but all data that flows through the data host is encrypted, including that received from the key party. Given the data host does not know query values, the data host cannot learn anything that it did not already possess prior knowledge about (e.g., data submitted using the data owner role). Hence, we argue this scheme can be simply extended to multiple data owners in this threat model. Do note that this scheme is not secure under the DO-Q threat model since it still computes exact k-NN.
 
The data host with the querier role appears more problematic due to information leakage. Intuitively, the ability to observe plaintext query outputs can leak the classification of data tuples. In particular, the data host can observe what tuples are used in the output construction, and the query output associates classifications to those tuples. It is challenging to defend against this breach because in outsourced exact k-NN, some cloud party typically must determine the subset of encrypted data tuples associated with the system output. If that party colludes with a querier (or possesses the querying role), the group will be able to associate those tuples with query-returned classifications. In this situation, our Gaussian kernel density algorithm can again provide privacy, since the query output is dependent on the classification over all encrypted tuples. Hence, the query output cannot be associated with any particular set of data tuples.

\subsection{Key Party Threat Models}

Finally, we consider a key party either colluding with or possessing the role of a data owner or a querier, but not both. Again we will consider extending the Paillier-based scheme in  \cite{arxiv-knn}. A key party with either roles should not be able to compromise this system even with multiple data-owners. K-NN computation is done with relations to distances from a query point. Without knowing the query point, the key party who knows values in the data set (through the data owner role) will not be able to associate observed distance-based values with known data. Furthermore, the values it observes should either be masked (as in our scheme), or encrypted (as in \cite{arxiv-knn}). A key party with querying ability will only observe masked or encrypted values, hence it will not be able to learn anything from the computation. 

\section{Conclusion}

In this paper, we have presented the first exploration of privacy preservation in an outsourced k-NN system with multiple data owners. Under certain threat conditions, we can extend existing single data owner solutions to secure the multi-data owner scenario. However, under a particularly practical threat model, exact k-NN cannot be secured due to a set of adaptive distance-learning attacks. This highlights the need for investigation into the inherent leakage from outputs of machine learning algorithms. As an alternative solution, we propose use of a Gaussian kernel in kernel density estimation, which is the family of algorithms k-NN belongs to. We present a privacy preserving system that supports this similar algorithm, as well as evidence of its similarity to k-NN. 
Admittedly, the scheme may not be practical for large data sets, given that the computational and communication complexities scale linearly.
Improving performance is a remaining challenge, and we hope this first step lays the groundwork for future optimized privacy-preserving solutions.

\section{Acknowledgements}

This work is partially supported by a National Science Foundation grant (CNS-1237265). The first author is supported by the National Science Foundation Graduate Research Fellowship. Any opinion, findings, and conclusions or recommendations expressed in this material are those of the authors(s) and do not necessarily reflect the views of the National Science Foundation.

\newpage

\bibliographystyle{abbrv}
\bibliography{ppknn}

\appendix

\section{Kernel density estimation and regression}
\label{sec:background-kde-appendix}
Given $n$ samples $\{\mathbf{x}_1, \cdots, \mathbf{x}_n\}$ of a random variable $\mathbf{x}$, one can estimate the probability density function $p(\mathbf{x})$ with a kernel $K$ in the following way, called kernel density estimation:
\begin{align*}
  p(\mathbf{x}) &= \frac{1}{n} \sum_{i=1}^n K(\|\mathbf{x} - \mathbf{x}_i\|)
\end{align*}
where $\|\cdot\|$ is a norm. Given this estimate, classification can be determined as the most probable class:
\begin{align*}
  \mathcal{D}_C &= \{x_i\,|\,(x_i, y_i) \in \mathcal{D}, y_i = C\} \\
  p(\mathbf{x} | y = C, \mathcal{D}) &= \frac{1}{|\mathcal{D}_C|} \sum_{\mathbf{x}_i \in \mathcal{D}_C} K(\|\mathbf{x} - \mathbf{x}_i\|) \\
  p(y = C | \mathcal{D}) &= \frac{|\mathcal{D}_C|}{|\mathcal{D}|} \\
  p(y = C | \mathbf{x}, \mathcal{D}) &= \frac{p(\mathbf{x} | y = C, \mathcal{D}) p(y = C | \mathcal{D})}{p(\mathbf{x} | \mathcal{D})} \\
  &= \frac{p(\mathbf{x} | y = C, \mathcal{D}) p(y = C | \mathcal{D})}{\sum_{C'} p(\mathbf{x} | y = C', \mathcal{D}) p(y = C' | \mathcal{D})} \\
  &= \frac{\frac{1}{|\mathcal{D}_C|} \sum_{\mathbf{x}_i \in \mathcal{D}_C} K(\|\mathbf{x} - \mathbf{x}_i\|)
   \frac{|\mathcal{D}_C|}{|\mathcal{D}|}}{\sum_{C'} \frac{1}{|\mathcal{D}_C'|} \sum_{\mathbf{x}_i \in \mathcal{D}_C'} K(\|\mathbf{x} - \mathbf{x}_i\|)
   \frac{|\mathcal{D}_C'|}{|\mathcal{D}|}} \\
  &= \frac{\sum_{\mathbf{x}_i \in \mathcal{D}_C} K(\|\mathbf{x} - \mathbf{x}_i\|)}{\sum_{\mathbf{x}_i \in \mathcal{D}} K(\|\mathbf{x} - \mathbf{x}_i\|)} \\
	&\propto \sum_{\mathbf{x}_i \in \mathcal{D}_C} K(\|\mathbf{x} - \mathbf{x}_i\|) \\
	\arg\max_C p(y = C | \mathbf{x}, \mathcal{D}) &= \arg\max_C \sum_{\mathbf{x}_i \in \mathcal{D}_C} K(\|\mathbf{x} - \mathbf{x}_i\|)
\end{align*}
Therefore, to classify a particular point $\mathbf{x}$, we can sum the kernel values of the points which belong to each class and determine which class has the largest sum.

\section{Additional Security Analysis}
\label{sec:additional_security_analysis}

In this section, we additionally analyze our scheme to show that it leaks no further data beyond what the classification output leaks, under our semi-honest adversary model. More formally, our scheme leaks no further data than an ideal scheme where the data is submitted to a trusted third party, who computes the Gaussian kernel value sums and produces the same classification output. In the below proofs, we abbreviate the data host as the DH and the crypto-service provider as the CSP. 
Also, our initial discussion of the protocol design in Section \ref{sec:protocol_design} already justified the security provided by each step of the construction. We will avoid repeating these detail in the proofs and analyze leakage by the algorithm as a whole.

\begin{theorem}
    \label{proof:sqdist}
    Algorithm \ref{alg:sqdist} leaks nothing about data tuples or queries.
\end{theorem}
\begin{proof}
As we see in step 1 of Algorithm \ref{alg:kernel}, the inputs to Algorithm \ref{alg:sqdist} are one encrypted data tuple and one encrypted query. In step 1 and 3 of the algorithm, only encrypted data is visible to the DH. In step 2, the CSP does see plaintext. However, each plaintext has had a random additive mask applied to the original true value, done by the DH in step 1(c). Thus, the CSP cannot determine the true value from the plaintext. Under the security of the Paillier cryptosystem, the DH and the CSP learn nothing about data tuples or queries.
\end{proof}

\begin{theorem}
    \label{proof:kernel}
    Algorithm \ref{alg:kernel} leaks nothing about the data tuples or queries.
\end{theorem}
\begin{proof}
    In step 1, the DH only operates on encrypted values. Theorem \ref{proof:sqdist} proves that the $SquareDist$ algorithm (Algorithm \ref{alg:sqdist}) leaks nothing about the data tuples or queries, so the DH learns nothing from running $SquareDist$.

    In step 2, the CSP again views plaintext, but a random additive mask has been applied (step 1(b)). Thus, the CSP cannot determine the original unmasked value, and does not learn anything about the data tuples or queries.
\end{proof}

\begin{theorem}
    \label{proof:class}
    Algorithm \ref{alg:class} leaks nothing about the data tuples or queries, except what may inherently leak from the classification output of kernel density estimation using a Gaussian kernel.
\end{theorem}
\begin{proof}
Step 0 leaks no information based on Theorem \ref{proof:kernel}.
Again, the DH only operates on encrypted values in step 1. In step 2, the CSP has access to the Gaussian kernel value for a masked value (from step 2(a) of Algorithm \ref{alg:kernel}), and a randomly masked classification vector. Since all values are masked, the CSP cannot determine the true value or classification of any data tuple or query. Note the CSP does not return a decrypted classification vector to the DH, resulting in the DH still only operating on encrypted values in steps 3 through 6. Steps 6 through 8 are simply Yao's garbled circuit protocol for the defined function $f$. Through this construction, the garbled circuit's inputs, which are masked values and their associated additive masks, are not revealed to any party that doesn't possess them. Only the DH possesses the masks, and only the \csp\ can obtain the masked kernel value sums through decryption. Since these \csp-owned values are masked, the \csp\ learns nothing from the plaintext. By the security of Yao's protocol, nothing else is leaked by the garbled circuit protocol except the circuit output, which is the classification output of our algorithm.

Thus, no information is revealed about the data tuples or queries from Algorithm 3 except what may inherently leak from the classification output of kernel density estimation with a Gaussian kernel. Our protocol is as secure as an an ideal model, where a single trust cloud party receives all data tuples and computes kernel density estimation in the clear.
\end{proof}

\section{Additional Considerations with Gaussian Kernel Density Classification}
\label{sec:practical_differences}

Here we discuss additional considerations with using Gaussian kernel density estimation in place of k-NN.
We construct hypothetical data sets where we can expect them to behave differently, and discuss the underlying causes for the discrepancy.
We also discuss some practical issues specific to using a Gaussian kernel and how to solve them.

\subsection{Causes of Divergence}
\begin{figure}
\centering
	\includegraphics[width=0.3\textwidth, height=30mm, keepaspectratio=true]{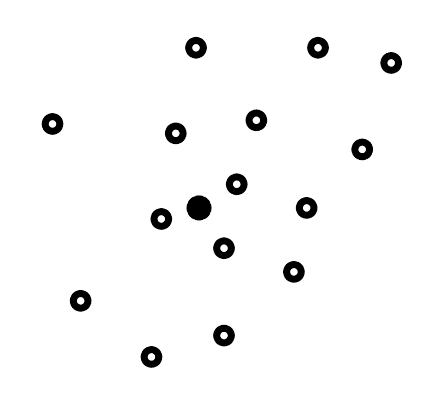}
	\caption{A visualization of unbalanced data. The filled points and unfilled points represent separate classes. Depending on the distance metric and the parameter of the Gaussian kernel, it is possible for classification with kernel density estimation to predict \emph{all} possible points in the space as the unfilled class, whereas $1$-NN classification would not.}
	\label{fig:unbalanced-data}
\end{figure}
Given a set of training data and parameters for the Gaussian kernel, assume query $q$ is classified as class $A$.
If sufficiently many points of class $B$ are added to the training data, the classification of $q$ will change to $B$ no matter how far away these extra points are.
In contrast, the classification under k-NN would remain unaffected if the new points are farther away than any of the existing $k$ nearest neighbors.
Figure~\ref{fig:unbalanced-data} illustrates a similar situation, where a point of one class is surrounded by many points of a different class.

To see why, recall from Appendix \ref{sec:background-kde-appendix}:
\[p(y = C | \mathcal{D}) = \frac{|\mathcal{D}_C|}{|\mathcal{D}|},\]
which means that the prior probability of a class is equal to its proportion in the training data.
If the data contains significantly more of one class than the other, then Gaussian kernel density classification will tend to predict the more frequent class, which arguably is preferable if the query points also satisfy this assumption.
If it is important to detect some less commonly occurring classes even at the cost of increased false positives, the less common classes should be weighted higher while computing the summed influences for each class. This can be done using a scaled basis vector for that class.

\subsection{Importance of Distance Metric}
A significant difference between classification with k-NN and with Gaussian kernel density estimation arises from the usage of the distances in computing the decision.
With k-NN, only the ordering of the distances to the labeled data points matters, and changes in the actual values of the distances have no effect if the order remains constant.
However, the Gaussian kernel is a functions of distance.
While the Gaussian kernel is monotonic (it decreases in value as distance increases), the sum of the values for each class can dramatically change even if the ordering of the distances between points does not.
In particular, as the Gaussian kernel is not linear, the classification of a query might change if distances are all scaled by some factor.

The limited precision of numbers used in calculation creates a more practical concern that the kernel value for a large distance will round down to $0$, effectively imposing a finite width on the Gaussian kernel.
If \emph{all} distances between pairs of points exceed this width, then kernel density estimation fails to provide any information about a query point; if a large fraction does, then the accuracy will correspondingly suffer.

To solve these issues, we recommend re-scaling all features to fit in $[0, 1]$.
Since the data owners already need to agree on the number of features and a common meaning for each feature, they can also find a consistent way to scale each feature.
In fact, if all features contain roughly equal information, then there is no particular reason some features should span a wider range and have greater influence on the distance metric.
Other domain-specific methods for ensuring that the distance between two points is never too large would also suffice.

\subsection{Selecting Parameters}
Both k-NN and the Gaussian kernel contain parameters we need to select: $k$, the number of neighbors to consider, and $\sigma$, the standard deviation for the Gaussian distribution, respectively. 
Choosing them inappropriately can lead to poor classification performance.
For example, if points of one class are scattered amongst points of another class, then a large $k$ will prevent us from classifying the scattered class correctly.
If $\sigma$ is too large, then the differences in the summed influence for each class will depend more heavily on the number of data points in each class rather than the distance to the points, and the system will tend to classify queries as the more frequent class.
Parameters should be selected using standard techniques such as cross-validation and grid search.


\end{document}